\pdfoutput=1

\def\year{2021}\relax
\documentclass[letterpaper]{article} 
\usepackage{aaai21}  
\usepackage{times}  
\usepackage{helvet} 
\usepackage{courier}  
\usepackage[hyphens]{url}  
\usepackage{graphicx} 
\usepackage{natbib}  
\usepackage{caption} 
\usepackage{graphicx}
\usepackage[dvipsnames]{xcolor}
\usepackage{cite}

\usepackage{amsfonts}       
\usepackage{booktabs}       
\usepackage{epsfig}         
\usepackage{float}          
\usepackage{graphicx}       
\usepackage{hyperref}       
\usepackage{mathtools}      
\usepackage{microtype}      
\usepackage{nicefrac}       
\usepackage{caption}
\usepackage{subfig}
\usepackage{url}  

\usepackage{xspace}
\usepackage{amssymb}

\usepackage{multirow}
\usepackage{enumitem}

 \usepackage[ruled,vlined,linesnumbered,noresetcount]{algorithm2e}
 \usepackage{algpseudocode}

\usepackage{wrapfig}

\usepackage{thmtools,thm-restate}
\usepackage{cleveref}

\usepackage{amsmath}
\usepackage{amssymb}
\usepackage{amsthm}
\usepackage{graphicx}
\usepackage{hyperref}
\usepackage{comment}

\newtheoremstyle{break}
  {\topsep}{\topsep}%
  {\itshape}{}%
  {\bfseries}{}%
  {\newline}{}%

\newtheorem{theorem}{Theorem}
\newtheorem{definition}{Definition}

\theoremstyle{break}
\newtheorem{experiment}{Experiment}

\title{Improving Robustness to Model Inversion Attacks via Mutual Information Regularization}

\author{
        Tianhao Wang\textsuperscript{\rm 1}\thanks{Equal Contribution},
        Yuheng Zhang\textsuperscript{\rm 2}\footnotemark[1],
        Ruoxi Jia\textsuperscript{\rm 3} \\
}
\affiliations{
    \textsuperscript{\rm 1} Harvard University \\
    \textsuperscript{\rm 2} University of Illinois Urbana-Champaign \\
    \textsuperscript{\rm 3} Virginia Tech \\
    tianhaowang@fas.harvard.edu, 
    yuhengz2@illinois.edu, 
    ruoxijia@vt.edu
}

\long\def\ignorethis#1{}

\usepackage{bbm}
\usepackage{amsthm}

\begin{document}

\maketitle

\begin{abstract}
This paper studies defense mechanisms against model inversion (MI) attacks -- a type of privacy attacks aimed at inferring information about the training data distribution given the access to a target machine learning model. Existing defense mechanisms rely on model-specific heuristics or noise injection. While being able to mitigate attacks, existing methods significantly hinder model performance. There remains a question of how to design a defense mechanism that is applicable to a variety of models and achieves better utility-privacy tradeoff.

In this paper, we propose the \textbf{M}utual \textbf{I}nformation Regularization based \textbf{D}efense (MID) against MI attacks. The key idea is to limit the information about the model input contained in the prediction, thereby limiting the ability of an adversary to infer the private training attributes from the model prediction. Our defense principle is model-agnostic and we present tractable approximations to the regularizer for linear regression, decision trees, and neural networks, which have been successfully attacked by prior work if not attached with any defenses.
We present a formal study of MI attacks by devising a rigorous game-based definition and quantifying the associated information leakage. 
Our theoretical analysis sheds light on the inefficacy of DP in defending against MI attacks, which has been empirically observed in several prior works. Our experiments demonstrate that MID leads to state-of-the-art performance for a variety of MI attacks, target models and datasets.

\end{abstract}

\section{Introduction}
Machine learning (ML) techniques have revolutionized many fields, such as computer vision, natural language processing, and robotics. However, the application of ML to domains involving sensitive and proprietary datasets is currently hindered by potential privacy threats. Recent studies have demonstrated various privacy attacks, which can expose the information about private training data through the access to a target model. The access could either be whitebox or blackbox. In the whitebox setting, the adversary has complete knowledge of the target model, whereas in the blackbox setting, the adversary is just allowed to make prediction queries against the model. Both settings can find concrete examples in today's ML-as-a-service platforms, such as Tensorflow Hub which offers a library of models that users can download and Microsoft Azure Cognitive Services which produce predictions for user-input data.

This paper focuses on MI attacks, a type of privacy attacks aimed at reconstructing the input associated with any given model output. MI attacks have been used to recover images of any person from a face recognition model given just
their name~\citep{zhang2019secret} and infer the genetic markers of individuals based on the medicine dosage prediction and some demographic information~\citep{fredrikson2014privacy}. 

Existing defense mechanisms against MI attacks can be categorized into two threads. One thread studies model-specific heuristics to mitigate attacks. For example, for decision trees, \citet{fredrikson2015model} studied the relationship between a sensitive feature's depth in the tree and the model's susceptibility to MI attacks and provided guidance for placing the sensitive feature in the tree. Although simple and efficient to implement, these heuristics cannot be easily generalized to a broader class of models and often only provide very limited protection against attacks. Another thread studies generic defense strategies that can be applied to any models. 
An example of such strategies is to train differentially private ML models.
However, prior work~\citep{fredrikson2014privacy} has empirically shown that DP can mitigate the MI attack only when the injected noise is large enough and as a side effect, the model suffers significant performance degradation. 
There still remains a question of how to design a defense mechanism that is applicable to a variety of models and achieves better utility-privacy tradeoff.

In this paper, we present a defense mechanism against MI attacks, called MID, which can achieve the state-of-the art performance for a variety of target models, datasets, and attack algorithms in both blackbox and whitebox settings. Since MI adversaries exploit the correlation between the model input and output for successful attacks, our idea for the defense is to limit the redundant information about the model input contained in the prediction. Algorithmically, we introduce a regularizer into the training loss, which penalizes the mutual information between the model input and the prediction. We present tractable approximations of the regularizer for all the models that have been successfully attacked before, which include linear regression, decision trees and neural networks. 

In addition, we provide a formal study of MI attacks and defenses. Particularly, existing theoretical framework for MI attacks focuses only on the privacy implications to training data while neglecting the fact that the attack also breaches privacy of other data from the same distribution. We take a first step toward formalizing the distributional privacy risk of MI attacks. 
Our theoretical analysis provides insights into the phenomenon of the inefficacy of DP in defending MI attacks observed in previous works.
We evaluate our defense mechanism for various target models, datasets, and attack algorithms, and demonstrate the superiority of our defense to the existing methods in terms of utility-privacy tradeoff.

\section{Related Work}

\paragraph{Attack Algorithms.} The first MI attack was demonstrated in \citep{fredrikson2014privacy}, where the authors presented a general algorithm to recover training data associated with an arbitrary model output given the model and some auxiliary information about the training set. The general idea of the algorithm is to formulate the MI attack as an optimization problem seeking for the sensitive input that achieves the maximum likelihood or posterior probability for the given model output and auxiliary information. \citet{fredrikson2014privacy} applied the algorithm to recover genetic markers given the linear regression model that uses them as part of input features. They also found that MI attacks are able to recover sensitive attributes for not only training data but also test data drawn independently from domain distribution.
\citet{fredrikson2015model} discussed the application of the general MI attack algorithm to more complex models including decision trees and some shallow neural networks. \citet{zhang2019secret} proposed a whitebox attack algorithm that can distill generic knowledge from public data and leverage it to improve the realism of reconstructed images for deep neural networks.
\citet{yang2019adversarial} focused on the blackbox setting and proposed to train a separate model that swaps the input and output of the target model to perform MI attacks. 
\citet{salem2019updates} studied the blackbox MI attacks for online learning, where the attacker has access to the versions of the target model before and after an online update and the goal is to recover the training data used to perform the update. They proposed to train a shadow model to mimic the target model, and then trains a separate model that transforms the change of the target model output in two online learning iterations into the private attributes. 

\paragraph{Formalization of MI attacks. } Several recent works started to formalize MI attacks and study the factors that affect a model’s vulnerability from a theoretical viewpoint. 
\citet{wu2016methodology} characterized model invertibility for Boolean functions using the concept of influence from Boolean analysis; \citet{yeom2018privacy} formalized the risk that the model poses specifically to individuals in the training data and showed that the risk increases with the degree of model overfitting. Both works define the privacy loss of MI attacks as the extent of information leakage of training data exceeds that of test data. On the contrary, our paper recognizes the privacy loss suffered by the test data drawn from the same distribution as training data and presents a formalism of MI attacks that allows for the analysis of distributional privacy leakage. 

\paragraph{Defenses.} There are very few studies about defenses against MI attacks. One idea is to use DP~\citep{fredrikson2014privacy}. DP provides a theoretical guarantee for training data privacy, but is not meant to protect the entire distribution. \citet{zhang2019secret} and \citet{fredrikson2014privacy} observed through empirical studies that DP cannot provide protection against MI attacks with reasonable model utility. Our paper presents a theoretical analysis that can explain the inefficacy of DP. Other defense ideas are model-specific. For instance, \citet{fredrikson2015model} proposed to place sensitive features at a particular depth to improve the robustness of decision trees. There also exists defenses designed specifically for blackbox attacks on neural networks, such as injecting uniform noise to confidence scores~\citep{salem2019updates}, reducing their precision~\citep{fredrikson2015model} or dispersion~\citep{yang2020defending}. Our defense differs from existing methods in that our approach is model-agnostic and applicable to both whitebox and blackbox settings. Also, we will show that our approach can achieve better utility-privacy tradeoff than the existing methods.

\newcommand{\mechanism}{\mathcal{M}}
\newcommand{\analyst}{\mathcal{A}}
\newcommand{\maxinfo}{I_{\infty}}
\newcommand{\eeps}{e^{\epsilon}}

\newcommand{\classifier}{f}
\newcommand{\targetdist}{p}
\newcommand{\unidist}{u}

\newcommand{\tgtind}{\targetdist_{X_s} \times \targetdist_{X_{ns}, Y}}

\newcommand{\unirv}{U}
\newcommand{\expsem}{\mathbf{Exp^{SEM}}(\analyst, \mechanism, \targetdist, \tau)}
\newcommand{\expind}{\mathbf{Exp^{IND}}(\analyst, \classifier, \targetdist)}
\newcommand{\dpname}{}
\newcommand{\expinddp}{\mathbf{Exp^{IND\dpname}}(\analyst, \mechanism, \targetdist)}

\newcommand{\expsemonly}{\mathbf{ Exp^{SEM}}}
\newcommand{\expindonly}{\mathbf{ Exp^{IND}}}
\newcommand{\expinddponly}{\mathbf{ Exp^{IND\dpname}}}

\newcommand{\gain}{\texttt{gain}}

\newcommand{\gainsemonly}{\gain^{\textbf{SEM}}}
\newcommand{\gainsem}{\gainsemonly(\analyst, \mechanism, \targetdist, \tau)}

\newcommand{\gainindonly}{\gain^{\textbf{IND}}}
\newcommand{\gainind}{\gainindonly(\analyst, \classifier, \targetdist)}

\section{Defense via Mutual Information Regularization}
\label{sec:formulation}

\subsection{Problem Setup}

In MI attacks, an adversary, given the access to a target model $f$ trained to predict specific labels $Y$, uses it to infer the information about the training data distribution $X$. We denote the output of the target model as $\hat{Y}$, i.e., $\hat{Y} = f(X)$. In addition, the adversary may also have access to some auxiliary knowledge $T$ that facilitates the inference. Consider the MI attacks against a face recognition model that labels an image containing a face with an identity corresponding to the individual depicted in the image. The  attack goal is to recover a representative image for some target identity $y$ based on the access to the target model. Possible auxiliary knowledge could be corrupted or blurred face images \cite{yang2019adversarial, zhang2019secret}. 

For the face recognition example, both the recovery of a training image for the target identity in and that of a test image (i.e., the image that does not appear in the training set but drawn from the same distribution) would incur privacy loss to the individual. Hence, it is important to design defenses to ensure the privacy of the entire training distribution, rather than just the members in training set. The goal of our defense is to design an algorithm to train the target model $f$ on the data distribution $(X,Y)$ such that the access to the resulting model does not allow an adversary to infer the information about $P(X|\hat{Y}=y)$.

A naive defense is to produce a classifier where $\hat{Y}$ is independent of $X$. In this case, the adversary cannot learn anything about the input data distribution for a given label. However, this classifier would clearly be useless in practice. Hence, there is a tradeoff between privacy and model utility and we want a defense that can achieve the best tradeoff between the two.

\subsection{Algorithm}
\label{sec:ourapproach}

Intuitively, we will need to limit the dependency between $X$ and $\hat{Y}$ to prevent the adversary from inferring the training data distribution associated with a specific label. Our idea is to quantify the dependence between $X$ and $\hat{Y}$ using their mutual information $\mathcal{I}(X;\hat{Y})$ and incorporate it into the training objective as a regularizer. Specifically, our defense, which we call MID, trains the target model via the following loss function:
\begin{equation}
\label{eq:regularizer}
\min_{\classifier \in \mathcal{H}}
E_{(x, y) \sim \targetdist_{X,Y}(x, y)}[\mathcal{L}(y, f(x))]
+ \lambda \mathcal{I}(X, \hat Y)
\end{equation}
where $\mathcal{I}(X, \hat Y)=\int_\mathcal{X} \int_\mathcal{Y} p_{X, Y}(x, y) 
    \log (\frac{p_{X, Y}(x, y)}{p_X(x)p_Y(y)}) dydx$,  $\mathcal{L}(y, f(x))$ is the loss function for the main prediction task, and
$\lambda$ is the weight coefficient that controls the tradeoff between privacy and utility on the main prediction task.

To deconstruct the proposed regularizer, we re-write the mutual information as follows:
\begin{equation}
\mathcal{I}(X, \hat Y) = \mathcal{H}(\hat Y) - \mathcal{H}(\hat Y | X) 
\end{equation}
When $f$ is a deterministic model, $\mathcal{H}(\hat Y | X)=0$ and introducing the mutual information regularizer effectively reduces the entropy of the model output, i.e., $\mathcal{H}(\hat Y)$. When $f$ is stochastic, the regularizer will additionally promote the uncertainty of the model output for a fixed input, i.e., $\mathcal{H}(\hat Y | X)$. In practice, reducing $\mathcal{H}(\hat Y)$ encourages the model output to be more concentrated and different inputs to be mapped into the same or very similar outputs; increasing $\mathcal{H}(\hat Y | X)$ makes the output to have a larger variance for a given $x$. Both terms will make $X$ less likely to be determined from the $\hat{Y}$. Figure \ref{fig:confidence_dist} illustrates this ``two-pronged'' privacy protection implied by our mutual information regularizer for a stochastic neural network consisting of a mean and a variance network trained on FaceScrub dataset. 
Figure~\ref{fig:confidence_dist} (a) plots the confidence value on the label class versus the sum of confidence on other classes for all test data with a particular label. It demonstrates that the regularizer will lead to more concentrated model outputs. Figure~\ref{fig:confidence_dist} (b) shows the simulated distribution of confidence values for a particular image through models trained with MID. We can see an increased trend of prediction variance for this input with increased $\lambda$, i.e. the strength of the MID regularizer.

\newcommand{\scaleconf}{6cm}
\begin{figure}
  \centering
  \includegraphics[width=\columnwidth]{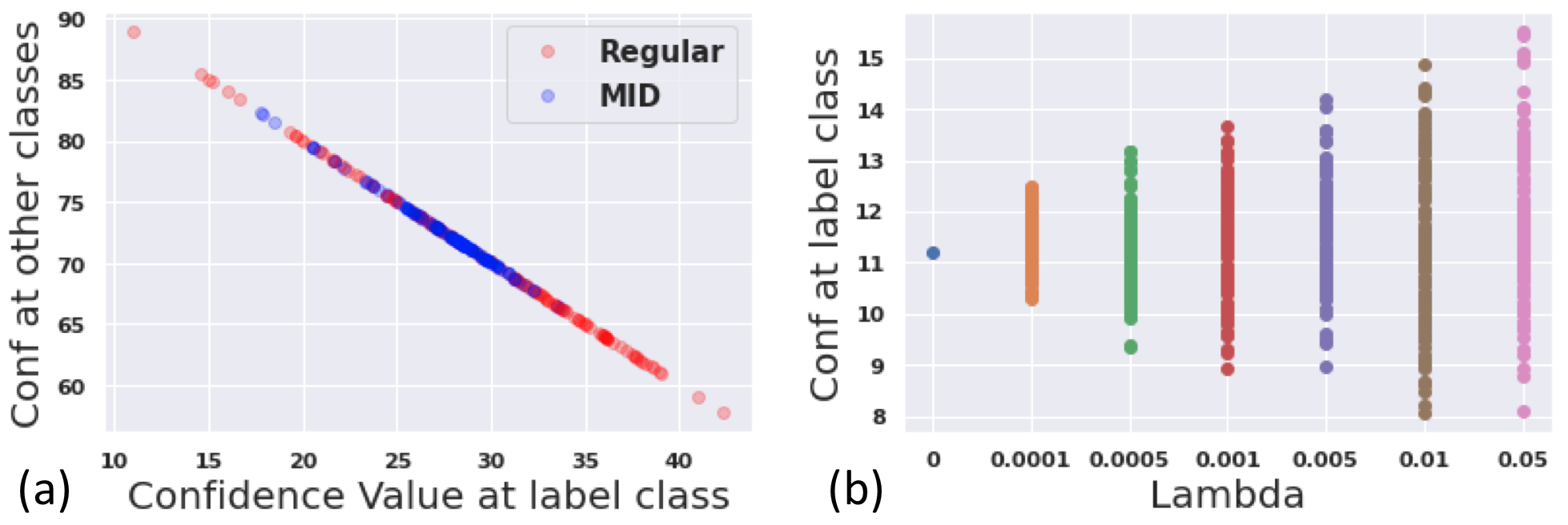}
  \caption{Illustration of the effect of penalizing $\mathcal{I}(X;\hat{Y})$ on the model output.}
  \label{fig:confidence_dist}
\end{figure}

\subsection{Instantiations of MID}
\label{sec:instantiation}
MID provides an appealing defense principle, as it defines what we mean by an effective defense, in terms of the fundamental tradeoff between reducing the input-output dependency and retaining good predictive power. However, computing mutual information is, in general, computational expensive. It requires modeling the joint distribution of model input and output and taking an integral over both domains, which is impracticable for most of the real-world prediction tasks. Here, we present efficient methods to implement the mutual information regularizer for different ML models that have been successfully attacked by previous work, including linear regression, decision trees, and neural networks. The unified idea underlying these methods is to find a tractable approximation to the mutual information regularizer.


\paragraph{Linear Regression.} 
A linear regression model can be written as $\hat y = Ax$. Due to the deterministic nature of the model, the mutual information regularizer is reduced to $\mathcal{H}(\hat Y)$. 
We proposed to approximate the distribution of $\hat Y$ by a Gaussian mixture: 
\begin{align}
    p(\hat y) = \frac{1}{N} \sum_{i=1}^N \mathcal{N}(\hat y | Ax_i; \sigma^2) 
\end{align}
where $\{x_i\}_{i=1}^N$ is the training set and $\sigma$ is a free parameter. We utilizes a Taylor-expansion based approximation for the entropy of Gaussian mixtures described in~\citet{huber2008entropy} and derive the following approximation to $\mathcal{I}(X, \hat Y)$:
\begin{align}
  &\mathcal{\tilde{I}}_\text{lin}(X, \hat Y) = \\
  &-\frac{1}{N} \sum_{i=1}^N \log \bigg(
\frac{1}{N} \sum_{i=1}^N \frac{1}{\sqrt{2\pi \sigma^2}} \exp{(-\frac{1}{2}(\frac{Ax_i - Ax_j}{\sigma})^2)}\bigg)\nonumber
\end{align}

\paragraph{Decision Trees.}
We modify the ID3~\citep{quinlan1986induction}, one of the most classic algorithms for training decision trees, to incorporate the mutual information regularizer. ID3 begins with the original training set as the root node. On each iteration of the algorithm, it iterates through every unused feature and calculates the splitting criterion (e.g. Information Gain or Gini Impurity), which measures the homogeneity of the labels within subsets. It then selects the feature achieving highest score based on the splitting criterion. The training set is then partitioned by the selected feature to produce subsets of the data. The algorithm continues to recurse on each subset and consider the features never selected before. In the inference phase, the decision tree is used to classify new test cases by traversing the tree using the features of the datum to arrive at a leaf node and label the datum with the most common class of the examples in the leaf node. Since decision trees trained with ID3 are deterministic, the mutual information regularizer again reduces to $\mathcal{H}(\hat{Y})$. To defend against the MI attacks, we add $-\mathcal{H}(\hat{Y})$ into the splitting criterion. Specifically, let $A$ denote a feature and $\mathcal{C}(A)$ denote some homogeneity measure. We select the feature $A$ that maximizes $\mathcal{C}(A) - \lambda \mathcal{H}(\hat{Y})$ to split on, where $\mathcal{H}(\hat{Y})$ is evaluated empirically with the training set. 



\paragraph{Neural Networks.} To come up with a tractable approximation to $ \mathcal{I} (X;\hat{Y})$ for neural networks, we get inspiration from the line of work on information bottleneck~\citep{shwartz2017opening,alemi2016deep} and regard the neural network as a Markov chain $Y-X-Z-\hat{Y}$, where $X$ is the feature, $Y$ is the ground truth label, $Z$ is a stochastic encoding of the input $X$ at some intermediate layer and defined by $P(Z|X;\theta)$, and $\hat{Y}$ is the prediction. By the data processing inequality~\citep{cover1999elements}, we have $\mathcal{I}(X, \hat Y) \le \mathcal{I}(X, Z)$. Prior work has provided various efficient approximation to $\mathcal{I}(X, Z)$~\citep{alemi2016deep,kolchinsky2019nonlinear}. Thus, we replace $\mathcal{I}(X, \hat Y)$ with its upper bound $\mathcal{I}(X, Z)$ in the training objective and train the neural network with the following loss function:
\begin{align}
    \min_\theta - \mathcal{I}(Z;Y) + \lambda \mathcal{I}(Z,X) 
\end{align}
The first term encourages the learned encoding to be maximally informative about the label $Y$ and measures the prediction performance of the model. The second term reduces the dependency between $X$ and $\hat{Y}$ by minimizing its upper bound and improves the robustness against the MI attacks. Note that the training objective above boils down to the classic information bottleneck and we will employ the variational method~\citep{alemi2016deep} to approximate the mutual information terms in the training objective.

\section{Theoretical Analysis}
In this section, we will formalize the MI attacks and quantify its \emph{distributional} privacy loss. Then, we will try to provide a theoretical basis for the recent observation that DP -- a canonical privacy notion nowadays -- cannot provide protection against the MI attacks with reasonable model utility~\citep{zhang2019secret,fredrikson2014privacy}.

\subsection{Formalizing MI Attacks}
\label{sec:formalizingmi}

We present a methodology for formalizing the MI attacks. Unlike previous works that capture the privacy loss of \emph{members in the training set} \cite{wu2016methodology, yeom2018privacy}, this is the first attempt of modeling the privacy loss of \emph{members in the population}. 

We have been mainly focused on the type of attacks that aim to recover all dimensions of the feature vector corresponding to a given label. Prior work on the MI attacks has also considered other types of attacks that only aim to reconstruct partial dimensions of the feature vector and have access to the remaining dimensions as auxiliary knowledge. For instance, \citet{fredrikson2014privacy} studied the MI attacks against a linear regression model that predict the medicine dosage based on the input of genotypes and some nonsensitive data like demographic information. To subsume this attack setting under our general formalization, we let the input random variable consists of the sensitive and nonsensitive part, i.e., $X=(X_s,X_{ns})$, and the attacker has access to $X_{ns}$. 

The attacker could also be interested in some specific property of the feature vector instead of the entire dimensions. In the face recognition example, the attacker could be interested in learning about the properties of a face image, such as hair color and race, instead of reconstructing the entire face image~\cite{zhang2019secret}. To capture this common attack scenario, we introduce $\tau$ to denote the property function that maps the sensitive feature to the property of the interest to the attacker.

We abstract the MI attacks into the following semantic MI game played between the server and the adversary. The game is described by a tuple $(\analyst, \mechanism, \targetdist, \tau)$, where $\analyst$ denotes the adversary which is a probabilistic machine with access to the learning algorithm $\mechanism$ and $p$ is the target data distribution. $\mechanism: \mathcal{X}^n \rightarrow \mathcal{H}$ is a learning algorithm that takes a dataset as input and outputs a single, fixed classifier $f$. $\mathcal{H}$ represents the hypothesis class.
\begin{experiment}
[Semantic MI $\expsemonly(\analyst, \mechanism, \targetdist, \tau)$]
\begin{enumerate}
    \item[]
    \item Server draws a training set $S \sim p^n$, and trains a classifer $f \leftarrow \mechanism(S)$. 
    \item Server draws $z=(x_{ns},x_s, y) \sim \targetdist$. $(f, x_{ns}, y)$ is presented to the adversary $\analyst$. 
    \item The adversary outputs $\analyst_{x_{ns}, y}$. $\expsem$ is $1$ if $\analyst_{x, y} = \tau(x_s)$, and $0$ otherwise. 
\end{enumerate}
\end{experiment}

The \emph{gain} of the game is evaluated as 
\begin{align}
\gainsem &= Pr[\expsem=1] \nonumber \\
&= Pr[
\analyst_{x_{ns}, y} = \tau(x_s)
]
\end{align}
where the probability is taken over the randomness of $S \sim p^n$, the randomness of $\mechanism$, the randomness of $\analyst$ and the randomness of $(x_s,x_{ns}, y) \sim \targetdist$. 
This game directly formalizes the procedure of MI attack. 
Trivially, for this game, the best strategy for the adversary is always output the most probable $\tau(x_s)$ when $x_s \sim p_{X_s|x_{ns}, y}$. 
Hence, the best possible gain for this game is 
\begin{equation}
E_{(x, y) \sim \targetdist}[ \max_v Pr_{x_s \sim p_{X_s|x_{ns}, y} }[\tau(x_s) = v] ]
\end{equation}


\newcommand{\advsem}{\texttt{Adv}^{\textbf{SEM}}}

To properly measure the adversary's advantage gained from the classifier $f$, 
we define an alternative underlying distribution $p_{X_s} \times p_{X_{ns}, Y}$ such that $X_s$ and $(X_{ns}, Y)$ have no correlation with each other. 
That is, sampling $x_s \sim p_{X_s}$ independently from $X_{ns}$ and $Y$. 
We 
define the advantage of $\analyst$ to be 
\begin{equation}
\begin{split}
     \advsem &(\analyst, \mechanism, \targetdist, \tau) 
     = \gainsem \\
     &- \gainsemonly(\analyst, \mechanism, \tgtind, \tau)
\end{split}
\end{equation}
i.e. $\gainsemonly(\analyst, \mechanism, \tgtind, \tau)$ is the trivial baseline gain of the adversary when there are indeed no correlation between $X_s$ and $(X_{ns}, Y)$, and $\advsem$ measures the extra gain the adversary can get. Note that there are multiple other ways to define the adversary's advantage which capture different insights. We defer this discussion to the Appendix. 

\subsection{Protection from Differential Privacy}
\label{sec:theory}


\newcommand{\invset}{f^{inv}} 
\newcommand{\major}{\mathrm{major}_\tau}
\newcommand{\majorspace}{\mathcal{T}_x^m}

\label{sec:theory-dp}
One dominate privacy notion is DP, which carefully randomizes an algorithm so that its output does not to depend too much on any single individual in the dataset \cite{dwork2014algorithmic}. 


\begin{definition}[Differential Privacy]
\label{def:dp}
Let $\mechanism: \mathcal{X}^n \rightarrow \mathcal{R}$ be a randomized mechanism. We say that $\mechanism$ is $(\epsilon, \delta)$-differentially private if for every two adjacent datasets $S \sim S'$ and every subset $R \subseteq \mathcal{R}$, 
\begin{align}
    Pr[\mechanism(S) \in R] \le e^\epsilon Pr[\mechanism(S') \in R] + \delta
\end{align}
\end{definition}


We introduce an indistinguishability game to derive the robustness guarantee offered by DP. The game is described by a tuple $(\analyst,\mechanism,p)$. 


\begin{experiment}[Indistinguishability$\dpname$~$\expinddp$]
\begin{enumerate}
    \item[]
    \item Server chooses $b \leftarrow \{0, 1\}$ uniformly at random. 
    \item Server draws training set $S \sim (p)^n$ if $b=0$; $S \sim (p_{X_s}\times p_{X_{ns},Y})^n$ if $b=1$. 
    \item Server trains a classifier $\classifier \leftarrow \mechanism(S)$ and presents $f$ to the adversary. 
    \item $\analyst$ outputs $b' \in \{0, 1\}$. $\expinddp$ is $1$ if $b=b'$, and 0 otherwise. 
\end{enumerate}
\end{experiment}

Let $\gainindonly \!(\analyst, \mechanism, \targetdist)\! = Pr[\expinddp=1]$. The following theorem shows that the gain of the semantic MI game can be upper bounded in terms of the best possible gain of this indistinguishability game. 
\begin{theorem}
For any attack strategy $\analyst^*$, learning algorithm $\mechanism$, target distribution $\targetdist$ and property function $\tau$,
\begin{multline}
\advsem\!(\analyst^*, \mechanism, \targetdist, \tau)\! \le 2\max_\analyst \gainindonly \!(\analyst, \mechanism, \targetdist)\! -\! 1 
\end{multline}

\label{thm:inddptosem}
\end{theorem}
Hence, to mitigate the threats of MI attack, we want $\max_\analyst \gainindonly$ to be not much greater than $1/2$.

We will provide a tight bound of $\gainindonly$ for any attack strategy $\analyst$ when the learning algorithm is differentially private. 
In the analysis, we assume that with high probability a training set drawn from $p^n$ has no intersection with another set drawn from $( \tgtind )^n$, i.e.,
\begin{equation}
Pr_{S \sim p^n, S' \sim ( \tgtind )^n}[S \sim_n S'] = 1 - \gamma
\end{equation}
where $S \sim_n S'$ indicates that the two datasets $S, S' \in \mathcal{X}^n$ differ by  $n$ entries, and $\gamma$ is a small value. 
This assumption is plausible for many practical scenarios where the feature vector has continuous domain or is high-dimensional. 


\begin{theorem}
\label{thm:dpguarantee}
If the learning algorithm $\mechanism: \mathcal{X}^n \rightarrow \mathcal{R}$ is $(\epsilon, \delta)$-differentially private, then with probability at least $1-\gamma$ we have tight bound $\max_\analyst gain^{IND\dpname}(\analyst, \mechanism, \targetdist) \le \frac{e^{n\epsilon}}{e^{n\epsilon}+1} + \frac{e^{n\epsilon}-1}{(e^{n\epsilon}+1)(e^\epsilon-1)} \delta$. 
\end{theorem}

As we can see, to make the upper bound $\frac{e^{n\epsilon}}{e^{n\epsilon}+1}
+ \frac{e^{n\epsilon}-1}{(e^{n\epsilon}+1)(e^\epsilon-1)} \delta = \frac{1}{2} 
+ \frac{e^{n\epsilon}-1}{2(e^{n\epsilon}+1)} 
+ \frac{e^{n\epsilon}-1}{(e^{n\epsilon}+1)(e^\epsilon-1)} \delta \in \frac{1}{2} + o(1)
$, 
the privacy budget $\epsilon$ needs to be set as $o(\frac{1}{n})$. However, this privacy budget is too small to allow any useful computation (see, e.g., Section 2.3.3 of \citet{dwork2014algorithmic}, for a comprehensive review). Since the bound in Theorem \ref{thm:dpguarantee} is tight, with high probability DP cannot mitigate the MI attacks with any reasonable model utility.


\begin{table*}[]
\centering
\resizebox{\linewidth}{!}{%
\begin{tabular}{c|c|c|c|c|c|c}
\toprule
&\textbf{Attack}     & \textbf{Model}    & \textbf{Dataset} & \textbf{Defense Baseline} & \textbf{Utility Metrics} & \textbf{Attack Metrics} \\ \midrule
\multirow{4}{*}{Blackbox} &Naive MAP           & Linear Regression & IPWC             & DP-AdaSSP                 & MSE                      & Acc; AUROC              \\ 
&Naive MAP           & Decision Tree     & FiveThirtyEight  & DPID3; Priority           & F1                       & F1                      \\ 
&Knowledge Alignment & Neural Networks   & FaceScrub        & DPSGD                     & Accuracy                 & Acc, L2, ECE            \\ 
&Update-Leaks        & Neural Networks   & CIFAR10          & DPSGD                     & Accuracy                 & Acc, ECE                     \\ \midrule
\multirow{2}{*}{Whitebox} &MAP with counts     & Decision Tree     & FiveThirtyEight  & DPID3; Priority           & F1                       & F1                      \\ 
&GMI                 & Neural Networks   & CelebA           & DPSGD                     & Accuracy                 & Acc, L2                 \\ \bottomrule
\end{tabular}
}
\caption{Summary of experimental settings.}
\label{tb:settingsummary}
\end{table*}

\section{Experiments}
In this section, we present the empirical evaluation of the efficacy of MID against different attacks and compare it with existing defense mechanisms. 

\subsection{Experiment Setting}

Table~\ref{tb:settingsummary} summarizes our experimental setting, including the attacks, models, datasets, metrics used in our evaluation as well as the other defenses that we compare with. We leave the description of datasets, attack hyper-parameters, model architectures as well as the details of the evaluation process and metrics to the Appendix. 

\paragraph{Attack algorithms.} We compare the efficacy of MID against the following MI attacks, which are the most effective ones presented in the literature thus far.

\begin{itemize}
    \item \underline{MAP}~\citep{fredrikson2014privacy,fredrikson2015model} casts the MI attack as an optimization problem that seeks for the maximum a posteriori probability (MAP) estimate of the sensitive attribute under the target model. 
    For decision trees, if the attacker also has access to the number of training examples for each path in the tree, it can further improve the attack performance by improving the MAP estimate 
    (referred as ``white-box with counts'' in \citep{fredrikson2015model}). We dub these two attacks \emph{Naive MAP} and \emph{MAP with counts}, respectively. 
    
    \item \underline{Knowledge Alignment}~\citep{yang2019adversarial} is a blackbox MI attack, in which the adversary trains an inversion model that swaps the input and output of the target model using the dataset drawn from a distribution similar to the private data distribution. The inversion model is then used to reconstruct the input feature for any given target model output.
    
    \item \underline{Update Leaks}~\citep{salem2019updates} is a blackbox MI attack designed for online ML models. It adopts a similar idea to~\citep{yang2019adversarial} and trains an inversion model to reconstruct the private data. The difference from~\citep{yang2019adversarial} is that the input to the inversion model now is the change of the target model output before and after an online update; moreover, the inversion model leverages the generative adversarial networks to improve the reconstruction accuracy.
    
    \item \underline{Generative MI (GMI)}~\citep{zhang2019secret} is a whitebox MI attack achieving the state-of-the-art performance against deep neural networks. GMI solves the MAP to recover the most possible private attribute via gradient descent. The key idea of GMI is to leverage public data to learn a generic prior for the private training data distribution and use it to regularize the optimization problem underlying the attack.

\end{itemize}

\paragraph{Defense baselines.} We compare our defense with DP for all attack algorithms, as well as other existing defenses available for specific models or attacks. We implement the AdaSSP, differentially private ID3, and DPSGD in~\citep{wang2018revisiting,friedman2010data,abadi2016deep} to construct differentially private linear regression models, decision trees, and neural networks, respectively. As for more specialized defenses, we compare with a defense proposed in~\citep{fredrikson2014privacy} for decision trees, which adjusts the depth of the sensitive features. We will use ``Priority'' to symbolize this defense. Previous work has also proposed to defend against blackbox attacks against neural networks by rounding the output confidence scores~\citep{fredrikson2015model} or injecting uniform noise to the confidence vector~\citep{salem2019updates}. However, our experiments found that these two defenses are not effective to protect against the attacks considered in this paper regardless of the rounding precision or the noise magnitude, respectively. Hence, we do not exhibit the results for these two defenses.

\paragraph{Evaluation Protocol.} We evaluate the performance of a defense mechanism in terms of the privacy-utility tradeoff. All the defenses considered in our paper have some hyperparameters which we could tune to vary the robustness and model performance. For the MID and DP, we can vary the weight parameter $\lambda$ in Equation~\ref{eq:regularizer} and the privacy budget $\epsilon$, respectively. For Priority, we can vary the depth of the tree where we place the sensitive features. In our experiment, we vary these hyperparameters and generate a utility-privacy tradeoff curve for each defense. We then use these curves to compare the performance of different defenses. When the target model is linear regression and a decision tree, we generate 100 models with different training and test data split and average the utility and privacy results over different models for each defense strategy and hyperparameter setting. When the target model is a neural network, we train 3 models and average the results. To illustrate the distributional privacy leakage, we demonstrate the attack performance on both training and test set for all attacks except GMI and Update-Leaks, since these two attacks aim at constructing a representative images for a given label instead of reconstruct a particular image in the training or test set. To compare the privacy risk of training and test data for GMI and Update-Leaks, we calculate the $L_2$ distances from deep feature representation of reconstructed images to that of training images and test images.



\paragraph{Evaluation Metrics.} The evaluation metrics for the model utility and attack performance are listed in Table \ref{tb:settingsummary}. The utility metrics depend on the underlying prediction task and datasets. For regression tasks, we employ \emph{mean squared error} (MSE); for classification tasks, we generally use \emph{test accuracy} as the metric unless the dataset is highly imbalanced in which case we use the \emph{F1 score}. For blackbox attacks on neural networks, there is a trivial defense to just release the prediction rather than the entire confidence vector and this defense achieves good utility in terms of test accuracy. However, this defense omits useful confidence information. To capture to what extent the defense worsen the confidence output, we leverage the \emph{Expected Calibration Error} (ECE)~\citep{guo2017calibration}, a metric commonly used to measure the miscalibration between confidence and accuracy. 
As for the attack performance, we follow the papers where the attacks were originally proposed and mostly use their metrics. Specifically, attack performance metrics for discrete attribute include \emph{accuracy}, \emph{AUROC}, and \emph{F1-score}.
The latter two are used when the sensitive attributes are highly imbalanced. To evaluate the attack performance for images, we build an evaluation classifier to see whether the image can be recognized as the target class and compute the \emph{attack accuracy}. For all the metrics above, the high value they take, the better attack performance. We also calculate the $L_2$ distances from deep feature representation of reconstructed images to that of training images and test images in order to compare the privacy threats to training and test data. 

\subsection{Results}
\begin{figure*}[ht!]
  \centering
  \includegraphics[width=\textwidth]{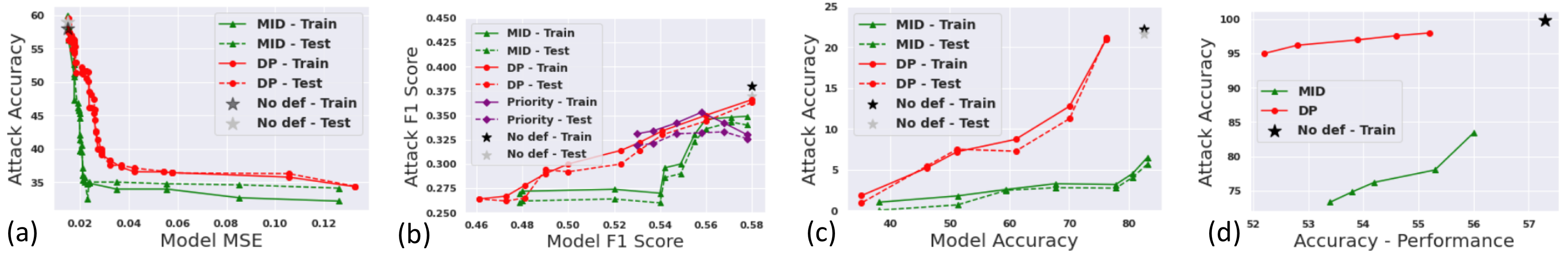}
  \caption{Defense result for blackbox MI attacks, including (a-b) Naive MAP on linear regression and decision trees, (c) Knowledge Alignment and (d) Update-Leaks on neural networks. }
  \label{fig:blackbox}
\end{figure*}

Figure~\ref{fig:blackbox} compares the performance of our defense against the existing baselines on various blackbox attacks, including Naive MAP (Figure~\ref{fig:blackbox} (a-b)), Knowledge Alignment (Figure~\ref{fig:blackbox} (c)), and Update-Leaks (Figure~\ref{fig:blackbox} (d)). For Naive MAP, we demonstrate the results on both linear regression and decision trees as these two models were successfully attacked by Naive MAP in prior work~\citep{fredrikson2014privacy,fredrikson2015model}. Knowledge Alignment and Update-Leaks are designed for attacking neural networks, so we only exhibit the results on neural networks for these two attacks. Figure~\ref{fig:blackbox} shows that MID can consistently achieve a better utility-privacy tradeoff than DP for protecting against all the attacks. For decision trees, there exists a model-specific defense strategy, named Priority, which adjusts the depth of the sensitive feature. Figure~\ref{fig:blackbox} (b) shows that MID is more robust than Priority in most cases except when model has a very high performance (i.e., F1 score $>$ 0.57). Nevertheless, Priority has its own drawbacks: firstly, it is only applicable to decision trees; and second, it only protects a subset of attributes, or the sensitive attributes, which could be subject to additional privacy concern because one may leverage the correlation between nonsensitive attributes and sensitive attributes to recover the sensitive one. 
On the other hand, the principles of MID and DP are applicable to different types of models and protect all the attributes at once. 
A phenomenon consistently present in all attacks is that the more predictive power the model has, the more vulnerable it is to the attacks, regardless of the types of defenses attached to the model. This finding signifies the difficulty to defend against the MI attacks; yet, our defense can significantly improve the model robustness for any fixed model performance. We also evaluate the attack performance in terms of other metrics shown in Table~\ref{tb:settingsummary}. As they give the similar results to the metrics exhibited in Figure~\ref{fig:blackbox}, we will omit these results to the Appendix.


\begin{figure}[ht!]
  \centering
  \includegraphics[width=\columnwidth]{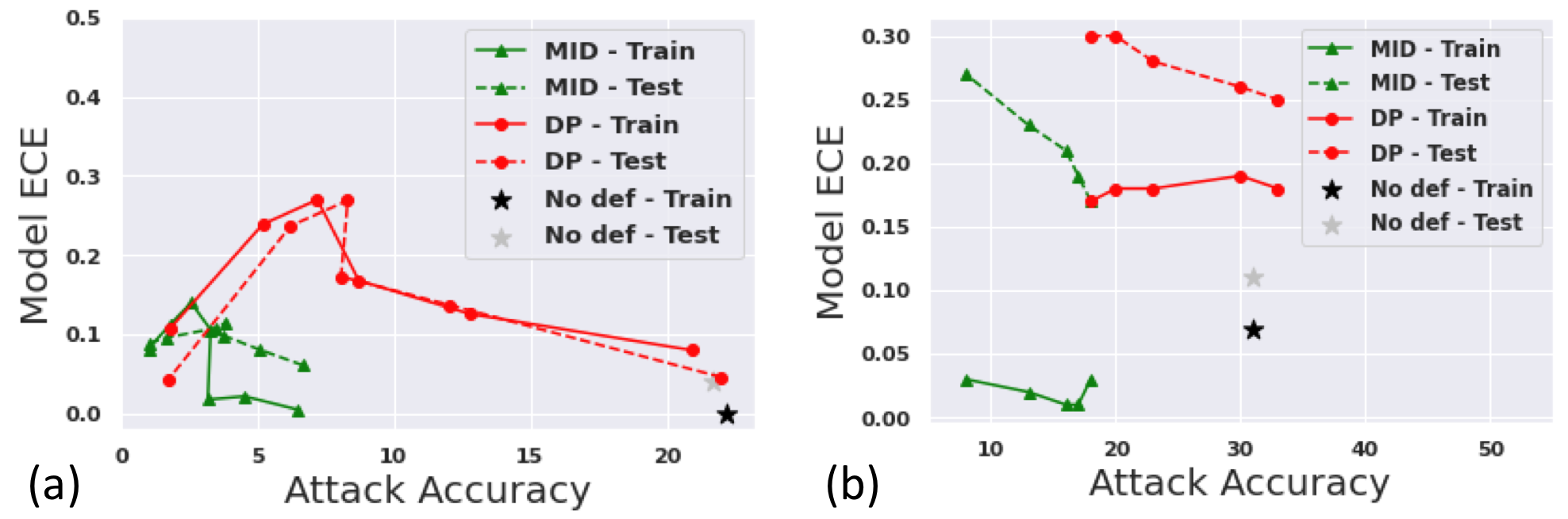}
  \caption{Confidence calibration of defenses evaluated on the (a) Knowledge Alignment and the (b) Update-Leaks attack. }
  \label{fig:blackbox-ece}
\end{figure}

Figure~\ref{fig:blackbox-ece} illustrates the miscalibration between confidence and accuracy as the result of different defenses. The degree of miscalibration is measured by ECE. We find that MID mostly achieves significantly lower ECE error than DP for a given robustness level.

\begin{figure}
\includegraphics[width=\columnwidth]{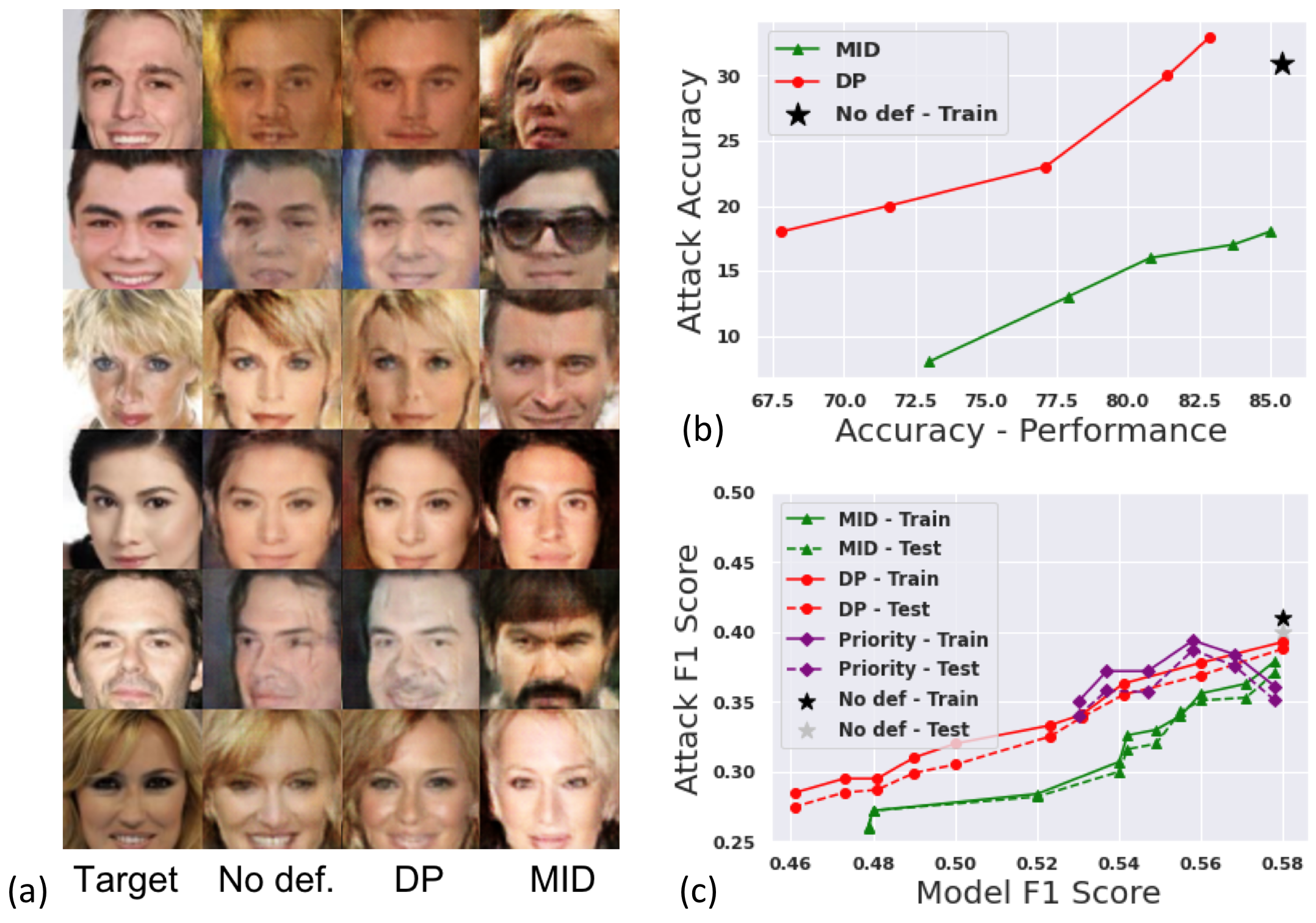}
\caption{Defense results for whitebox MI attacks, including (a-b) GMI on neural networks and (c) MAP with counts on decision trees.}
\label{fig:whitebox}
\end{figure}

Figure~\ref{fig:whitebox} illustrates the defense results for whitebox attacks. Specifically, we consider GMI, the state-of-the-art attack against neural networks, and 
MAP with counts, the only known attack against decision trees; the corresponding privacy-utility tradeoff curves are shown in Figure~\ref{fig:whitebox} (b) and (c), respectively. We find that MID can outperform DP by a large margin to defend against whitebox MI attacks. Similar to the observation in blackbox experiment, model-specific defense strategies could achieve better robustness than MID occasionally; but overall, MID could be more preferable due to its broad applicability, thorough protection for all attributes, and reliable performance. Figure~\ref{fig:whitebox} (a) compares the reconstructed images of MID and DP. We tune the hyperparameters of both defenses so that the resulting target models have comparable performance. The test accuracies for the models trained with MID and DP are 80.8 and 81.4, respectively. We can see that MID can block the attack much better than DP. For instance, the reconstructions for DP can still retain many facial features of the target individual while the reconstructions for MID are almost completely different from the target individual. 

One common observation from all defense results is that the attack performance is very close on  training and test set. This implies that MI attacks pose similar privacy threats regardless of whether an individual's data is selected for training or not. Although some prior works \cite{yeom2018privacy, dwork2015generalization} report that the attack performance on the training set is much higher than the test set, we conjecture that it is because the model overfits the training data in their evaluations. In our experiments, since both MID and DP can help mitigate overfitting, the privacy threats of MI attacks is presented to the population instead of the training set only.

\section{Conclusion}
We propose a defense against MI attacks based on regularizing the mutual information between the model input and prediction and further present tractable approximations to the regularizer for linear regression, decision trees, and neural networks. We provide theoretical basis for a common empirical observation that DP cannot defend against MI attacks with reasonable model utility. We perform experiments to compare the utility-privacy tradeoff of our defense against existing baselines on different models and datasets, and demonstrate that our proposed defense can achieve the state-of-the-art performance to protect against various attacks in both whitebox and blackbox settings. 
One limitation of our defense is that it does not prevent the attack but instead aim for better utility-privacy tradeoff. 
For future work, we would like to further improve our defense from the perspective of computational security, i.e., by assuming a polynomial time adversary.

\newpage

\bibliography{ref}

\newpage

\newcommand{\expsemonlyideal}{\mathbf{ Exp^{SEM}_{ideal}}}

\newcommand{\gainsemideal}{\gain^{\textbf{SEM}}_{\textbf{ideal}}}

\theoremstyle{break}

\section{Appendix A: Additional Experiment Details}

We give more details of datasets, evaluation metrics, as well as different attack experiments here. We also submit our code as part of the supplementary materials. 

\subsection{Datasets}
\paragraph{IWPC} The IWPC collected data on patients who were prescribed warfarin in 9 counties. 17 attributes of the patients are collected including a genotype VKORC1, which is considered as the target sensitive attributes in our paper. The outcome variable is the stable therapeutic dose of Warfarin. This dataset is considered as a benchmark dataset for linear regression inversion attack \cite{fredrikson2014privacy}. Note that in \cite{fredrikson2014privacy}, another genotype CYP2C9 is also considered as sensitive target, but since CYP2C9 has low correlation with dosage and demographic information, we show the attacking results on VKORC1 only. 

\paragraph{FiveThirtyEight} FiveThirtyEight is a survey of 553 individuals from SurveyMonkey which support the analysis of the connection between peoples' steak preparation preferences and their propensity for risk-taking behaviors. The response for the question about infidelity is considered sensitive, following the settings in \cite{fredrikson2015model}. 

\paragraph{FaceScrub} A dataset of 48,579 images for 530 individuals. We extract the face of each image according to the official bounding box information. Each image is resized to $64 \times 64$, following the settings in \cite{yang2019adversarial}. 

\paragraph{CelebA} the CelebFaces Attributes Dataset (CelebA) containing 202,599 face images of 10,177 identities with coarse alignment. We crop the images at the center and resize them to $64 \times 64$ so as to remove most background, following the settings in \cite{zhang2019secret}.

\paragraph{CIFAR10} A dataset consists of 60,000 images in 10 classes (airplane, automobile, bird, cat, deer, dog, frog, horse, ship and truck). Each image is resized to $64 \times 64$.


\subsection{Utility Metrics}

The utility metrics depend on the underlying prediction task and datasets. For regression tasks, we employ \emph{mean squared error} (MSE); for classification tasks, we generally use \emph{test accuracy} as the metric unless the dataset is highly imbalanced in which case we use the \emph{F1 score}. For blackbox attacks on neural networks, there is a trivial defense to just release the prediction rather than the entire confidence vector and this defense achieves good utility in terms of test accuracy. However, this defense omits useful confidence information. 

To evaluate the impact of defense algorithm on the precision of confidence information, we measure the degree of model calibration in terms of \emph{Expected Calibration Error} (ECE) for each different model \cite{guo2017calibration}. We say that the defense A dominates the defense B if A can achieve better robustness than B when their model performance is fixed to any value (in some reasonable range). 
Model calibration is the degree to which a model’s predicted probability estimates the true correctness likelihood. ECE measures the difference in expected accuracy and expected confidence. In practice, we group predictions into $M$ bins, and let $B_m$ be the set of indices of samples whose prediction confidence falls into the interval $I_m = ( \frac{m-1}{M}, \frac{m}{M} ]$ for $m \in \{1 \dots, M \}$.
ECE is then computed as 
$$
\sum_{m=1}^M \frac{|B_m|}{n} |acc(B_m) - conf(B_m)|
$$
where $acc(B_m) = \frac{\sum_{i \in B_m} I[\arg\max_y \hat y_i = y_i] }{|B_m|}$
and $conf(B_m) = \frac{1}{|B_m|} \sum_{i \in B_m} \hat p_i$, where $\hat p_i$ is the confidence for sample $i$ on label $y_i$. 

\subsection{Privacy Metrics}
As for the attack performance, we follow the papers where the attacks were originally proposed and mostly use their metrics. 
There are three metrics to measure the attack performance when the target sensitive attribute is discrete: (1) \emph{accuracy}, which is the percentage of samples for which the attack algorithm correctly predicted genotype, (2) \emph{AUROC}, which is the multi-class area under receiver operating characteristic (ROC) curve, and (3) \emph{F1-score}, which is the harmonic average of attack precision and recall. This is because in both IWPC and FiveThirtyEight the distribution of the target attribute is highly unbalanced. 
Although attack accuracy is generally easier to interpret, it can be misleading in this case.
AUROC and F1-score are intended to help evaluating when the sensitive attributes are highly imbalanced. 

For MI attacks on image recognition model, 
\emph{L2 Distance} is measured between the original and reconstructed images. 
We also build an evaluation classifier to evaluate the \emph{attack accuracy} on neural network models. If the evaluation classifier can correctly classify the reconstructed image, then the image is considered to be close to a representative image for a certain label class. The evaluation classifier is different from the target classifier and has very high performance, which serves as a label predicting oracle. For knowledge alignment experiment (FaceScrub), our evaluation network is adapted from VGG-16 \cite{simonyan2014very}. 
For GMI experiment (CelebA), we use the model in \cite{cheng2017know} which is pretrained on the MS-Celeb-1M~\cite{guo2016ms} and then fine-tuned on the training set of the target network.

\subsection{Attack Implementation \& Models}

\paragraph{Linear Regression}
We implement the MI attack against linear regression on IWPC dataset following the setting in~\citep{fredrikson2014privacy}. 

\paragraph{Decision Tree}
We implement the MI attack against decision trees using the algorithms presented in~\citep{fredrikson2015model}, including both black-box and white-box with count attack. We use the FiveThirtyEight dataset, following~\citep{fredrikson2015model}.

\paragraph{Knowledge Alignment}
In~\citep{yang2019adversarial}, a second neural network is trained as the inverse of the target model to perform the MI attack. The adversary's background knowledge is leveraged to compose an auxiliary set to train the inversion model.
Specifically, the target classifier $\classifier$ is trained using dataset $S$ drawn from the target distribution $\targetdist$. An inversion model $g$ is separately trained using another set $S_{aux}$ drawn from $\targetdist$, i.e. the \emph{auxiliary set}, where the output $f(x)$ for $x \in S_{aux}$ as the input and $x$ as the expected output. $g$ is then used to perform data reconstruction of $\classifier$'s training set and test set. 
To stage a successful attack, the attacker model often used is that the full prediction values $\hat y$ are given, and the adversary also has the access to auxiliary dataset $S_{aux}$ that is drawn from the exactly same distribution as the original training set $S$. Note that this assumption is overly strong as pointed out by the authors in the original paper. To obtain meaningful attack results and compare with \citep{yang2019adversarial}, our experiments follow this assumption. The architectures for the target classifier $f$ and the inversion model $g$ are exactly the same as the ones used in~\citep{yang2019adversarial}.  


\paragraph{Update-Leaks}
We use a simple CNN which consists of two convolutional layers, one max pooling layer, one fully connected layer, and a softmax layer trained on CIFAR10 as the target model. We implement the attack algorithm with the exact same settings as in \cite{salem2019updates}. 

\paragraph{Generative Model Inversion (GMI)}
We use VGG16 trained on CelebA dataset as the target model, and implement the GMI attack following the same settings in \cite{zhang2019secret}. 

\begin{figure}[t!]
  \centering
  \subfloat[]{
    \includegraphics[width=0.3\textwidth]{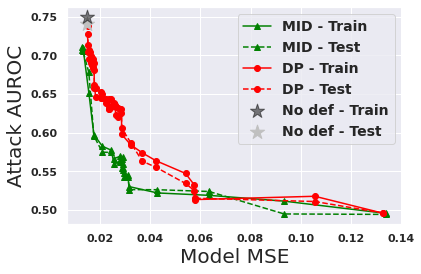}
    \label{fig:linearreg-auroc}
  }\\
  \subfloat[]{
        \includegraphics[width=0.3\textwidth]{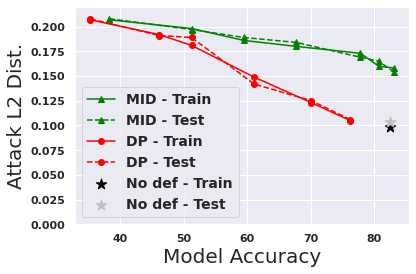} 
        \label{fig:knowledgealign-l2}
  } \\
  \subfloat[]{
        \includegraphics[width=0.3\textwidth]{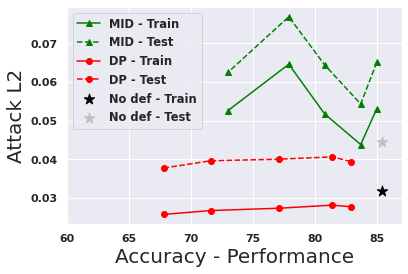} 
        \label{fig:gmi-l2}
  }
  \caption{Defense results for blackbox and whitebox MI attacks measured with additional metrics, including (a) Naive MAP against linear regression with AUROC as the attack performance metric; (b) Knowledge Alignment against a neural network with deep feature L2-distance as the attack performance metric (c) GMI on a neural network with deep feature L2-distance as the attack performance metric.}
  \label{fig:additional}
\end{figure}

\subsection{Additional Results} We present defense results for blackbox and whitebox MI attacks measured with additional metrics. Figure~\ref{fig:additional} shows that our proposed defense outperforms DP on these additional metrics. Comparing Figure~\ref{fig:additional} (b) with (c), we found that training data are subject to a higher privacy loss than test data for the model attacked by GMI, while the privacy loss of training and test data is about the same for the model attacked by Knowledge distillation. One potential reason is that the model in (c) has more parameters and thus is subject to more  overfitting. If the model overfits to training data, the model tends to learn the input-output correlation better for training data than test data and as a result, training data is more susceptible to MI attacks. Another potential reason is that the attack algorithm in (c) is more powerful as it is a whitebox attack while the attack in (b) only has blackbox access to the model. Hence, the attack in (c) can better reflect the subtle difference between training and test data in terms of their vulnerability. We leave the validation of these hypotheses to future work.


\section{Appendix B: Additional Theoretical Results}

The definition of differential privacy can be generalized to datasets that differ on more than a single entry. We say that two datasets $S, S' \in \mathcal{X}^n$ are \emph{$k$-adjacent} if they differ by $k$ entries, and we denote this by $S \sim_k S'$. 
\begin{theorem}[Group Privacy]
\label{thm:groupdp}
For every $k \ge 1$, if $\mechanism: \mathcal{X}^n \rightarrow \mathcal{R}$ is $(\epsilon, \delta)$-differentially private, then for every two $k$-adjacent datasets $S \sim_k S'$, and every subset $R \subseteq \mathcal{R}$, we have
$$
Pr[\mechanism(S) \in R] \le e^{k\epsilon} Pr[\mechanism(S') \in R] + 
\frac{e^{k\epsilon}-1}{e^\epsilon-1} \delta
$$
where the bound is tight.
\end{theorem}
Another central property of differential privacy is that it is preserved under post-processing. 
\begin{theorem}[Differential Privacy under Post-processing]
\label{thm:postprocessing}
If $\mechanism: \mathcal{X}^n \rightarrow \mathcal{R}$ is $(\epsilon, \delta)$-differentially private, and $\mathcal{F}: \mathcal{R} \rightarrow \mathcal{Z}$ is any randomized function, then $\mathcal{F} \circ \mechanism: \mathcal{X}^n \rightarrow \mathcal{Z}$ is $(\epsilon, \delta)$-deferentially private. 
\end{theorem}


\subsection{Discussion of the trivial baseline for MI attacks}
There are several possible ways to define the trivial baseline of the attack gain. A common practice in cryptography is to define a simulator $\mathcal{S}$ which has access to everything else but the target model $\classifier$: 
\begin{equation}
     \mathcal{S}(\analyst, \targetdist, \tau) = Pr_{\analyst, (x_s, x_{ns}, y) \sim \targetdist} [\analyst_{x_{ns}, y} = \tau(x_s)] 
\end{equation}
and the advantage of $\analyst$ can be defined as \begin{equation}
     \advsem_{\mathcal{S}} (\analyst, \classifier, \targetdist, \tau) 
     = \gainsem - \mathcal{S}(\analyst, \targetdist, \tau)
\end{equation}
Another potential baseline definition is to define an alternative underlying distribution $\tgtind$ such that $X_s$ and $(X_{ns}, Y)$ have no correlation with each other. 
That is, sampling $x_s \sim \targetdist_{X_s}$ independently from $X_{ns}$ and $Y$. 
We then define another classifier $f^*$ that is trained from a training set $S$ draws from $\tgtind$, and define the advantage of $\analyst$ to be 
\begin{equation}
\label{eq:advwrong}
\begin{split}
     \advsem (\analyst, \classifier, \targetdist, &\tau) = \gainsem \\
     &- \gainsemonly(\analyst, \mechanism, \tgtind, \tau)
\end{split}
\end{equation}
This definition captures the intuition that if $X_{s}$ has little correlation with $X_{ns}$ and $Y$, the advantage of the attacker should be small no matter which strategy he/she uses and which classifier $f$ it has access to. We can also define an alternative experiment: 
\begin{experiment}
[Ideal World $\expsemonlyideal (\analyst, \mechanism, \targetdist, \tau)$]
\begin{enumerate}
    \item[]
    \item Server draws a training set $S \sim (\tgtind)^n$, and trains a classifer $f \leftarrow \mechanism(S)$. 
    \item Server draws $z=(x_{ns},x_s, y) \sim \targetdist$. $(f, x_{ns}, y)$ is presented to the adversary $\analyst$. 
    \item The adversary outputs $\analyst_{x_{ns}, y}$. $\expsemonlyideal(\analyst, \mechanism, \targetdist, \tau)$ is $1$ if $\analyst_{x_{ns}, y} = \tau(x_s)$, and $0$ otherwise. 
\end{enumerate}
\end{experiment}
In this experiment, the classifier $f$ still provides correlation information between $X_{ns}$ and $Y$ but no longer contain information for the correlation between $X_s$ and $(X_{ns}, Y)$, thus it's information-theoretically hard to compute $\tau(x_s)$. Using this experiment, we can have another definition of the advantage of $\analyst$:
\begin{equation}
\label{eq:advideal}
\begin{split}
     &\advsem (\analyst, \classifier, \targetdist, \tau) \\
     &= \gainsem - \gainsemideal (\analyst, \mechanism, \targetdist, \tau)
\end{split}
\end{equation}

\subsection{Proof of Theorem 1}
\begin{theorem}
For any attack strategy $\analyst^*$, learning algorithm $\mechanism$, target distribution $\targetdist$ and property function $\tau$,
\begin{equation}
\advsem\!(\analyst^*, \mechanism, \targetdist, \tau)\! \le 2\max_\analyst \gainindonly \!(\analyst, \mechanism, \targetdist)\! -\! 1 
\end{equation}
for 
\begin{align}
\advsem (\analyst, \classifier, \targetdist, \tau) &= \gainsem \\
&~~~~- \gainsemideal (\analyst, \mechanism, \targetdist, \tau)
\end{align}
\end{theorem}

\begin{proof}
For any strategy $\analyst^*$ in for the semantic game $\expsem$, we can construct $\analyst$ for $\expinddp$ where for any trained classifier $\classifier$, $\analyst$ outputs $0$ if $\analyst^*$ wins on an instance $(x_s, x_{ns}, y) \sim p$. 
The gain of this strategy is 
\begin{align} 
    &\gainindonly(\analyst, \mechanism, \targetdist) = 
    Pr[\expinddp=1] \\
    &= Pr[\analyst(\mechanism(S))=0|b=0]Pr[b=0] \\
    &~~~~+ Pr[\analyst(\mechanism(S))=1|b=1]Pr[b=1] \\
    &= \frac{1}{2}Pr_{S\sim (p)^n, \mechanism, \analyst} [\analyst(\mechanism(S))=0] \\
    &~~~~+ 
    \frac{1}{2} Pr_{S \sim (\tgtind)^n, \mechanism, \analyst}[\analyst(\mechanism(S))=1] \\
    &= \frac{1}{2}Pr_{S\sim (p)^n, \mechanism, \analyst} [\analyst(\mechanism(S))=0] \\
    &~~~~+ 
    \frac{1}{2} [1-Pr_{S \sim (\tgtind)^n, \mechanism, \analyst}[\analyst(\mechanism(S))=0]] \\
    &= \frac{1}{2}Pr_{S\sim (p)^n, \mechanism,  \analyst^*} [\analyst^*_{x_{ns}, y} = \tau(x_s)] \\
    &~~~~+
    \frac{1}{2} [1-Pr_{S \sim (\tgtind)^n, \mechanism, \analyst^*}[\analyst^*_{x_{ns}, y} = \tau(x_s)]] \\
    &= \frac{1}{2} \gainsemonly(\analyst^*, \mechanism, \targetdist, \tau) \\
    &~~~~+ \frac{1}{2} (1-\gainsemideal (\analyst, \mechanism, \targetdist, \tau)) \\
    &= \frac{1}{2} + \frac{1}{2} \advsem(\analyst^*, \mechanism, \targetdist, \tau)
\end{align}
where in (15) and (16) we omit the randomness of $(x_s, x_{ns}, y) \sim p$ in the probability space due to space constraint.
By the above result, we have 
\begin{equation}
\begin{split}
\max_\analyst \gainindonly(\analyst, \mechanism, \targetdist) \ge 
\frac{1}{2} + \frac{1}{2} \advsem(\analyst^*, \mechanism, \targetdist, \tau)
\end{split}
\end{equation}
which holds for any attack strategy $\analyst^*$.
\end{proof}

\subsection{Proof of Theorem 2}

We will provide a tight bound of $\gainindonly$ for any attack strategy $\analyst$ when the learning algorithm is differentially private. 
In the analysis, we assume that with high probability a training set drawn from $p^n$ has no intersection with another set drawn from $( \tgtind )^n$, i.e.,
\begin{equation}
\label{eq:assumption}
Pr_{S \sim p^n, S' \sim ( \tgtind )^n}[S \sim_n S'] = 1 - \gamma
\end{equation}
where $S \sim_n S'$ indicates that the two datasets $S, S' \in \mathcal{X}^n$ differ by  $n$ entries, and $\gamma$ is a small value. 
This assumption is plausible for many practical scenarios where the feature vector has continuous domain or is high-dimensional. 


\begin{theorem}
If the learning algorithm $\mechanism: \mathcal{X}^n \rightarrow \mathcal{R}$ is $(\epsilon, \delta)$-differentially private, then with probability at least $1-\gamma$ we have tight bound 
\begin{equation}
\max_\analyst gain^{IND\dpname}(\analyst, \mechanism, \targetdist) \le \frac{e^{n\epsilon}}{e^{n\epsilon}+1} + \frac{e^{n\epsilon}-1}{(e^{n\epsilon}+1)(e^\epsilon-1)} \delta. \nonumber
\end{equation}
\end{theorem}
\begin{proof}
By the post-processing property in Theorem \ref{thm:postprocessing}, for any attack strategy $\analyst: \mathcal{H} \rightarrow \{0, 1\}$ we know that $\analyst \circ \mechanism $ is also $(\epsilon, \delta)$-differentially private. By group privacy property in Theorem \ref{thm:groupdp} and our assumption in \ref{eq:assumption}, with probability at least $1-\gamma$, dataset $S \sim \targetdist^n$ and $S' \sim (\tgtind)^n$ are differ by $n$ records, and for any subset $R \subseteq \{0, 1\}$ we have 
\begin{equation}
\begin{split}
&Pr_{\analyst, \mechanism, S \sim p^n}[\analyst(\mechanism(S)) \in R] \\
&\le e^{n\epsilon}
Pr_{\analyst, \mechanism, S' \sim (\tgtind)^n}[\analyst(\mechanism(S')) \in R] \\
&~~~~+ \frac{e^{n\epsilon}-1}{e^\epsilon-1} \delta
\end{split}
\end{equation}
Thus we have
\begin{align}
&Pr[\expinddp=1] \\
& = Pr[\analyst(\mechanism(S))=0 | b=0] Pr[b=0]
\\
&~~~~+ Pr[\analyst(\mechanism(S'))=1 | b=1] Pr[b=1] \\
&= \frac{1}{2}Pr_{S\sim p^n}[\analyst(\mechanism(S))=0] \\
&~~~~+ 
\frac{1}{2} Pr_{S'\sim (\tgtind)^n}[\analyst(\mechanism(S'))=1] \\
&\le \frac{1}{2}e^{n\epsilon}Pr_{S'\sim (\tgtind)^n}[\analyst(\mechanism(S'))=0] \nonumber \\
&~~~~+ \frac{1}{2}e^{n\epsilon}Pr_{S\sim p^n}[\analyst(\mechanism(S))=1] + \frac{e^{n\epsilon}-1}{e^\epsilon-1} \delta \\
&= e^{n\epsilon} Pr[\expinddp=0] + \frac{e^{n\epsilon}-1}{e^\epsilon-1} \delta \\
&= e^{n\epsilon} (1-Pr[\expinddp=1]) + \frac{e^{n\epsilon}-1}{e^\epsilon-1} \delta
\end{align}
Hence for any $\analyst$ we have $Pr[\expinddp=1] \le \frac{e^{n\epsilon}}{e^{n\epsilon}+1}
+ \frac{e^{n\epsilon}-1}{(e^{n\epsilon}+1)(e^\epsilon-1)} \delta
$.
\end{proof}

\end{document}